\documentclass[11pt]{article} 
\usepackage[margin=1in]{geometry} 
\pdfoutput=1 

\usepackage[usenames,dvipsnames,svgnames,table]{xcolor} 

\usepackage{amsmath,amsfonts,amssymb,amsthm}
\usepackage{mathtools}
\usepackage[linesnumbered,ruled,noline]{algorithm2e} 

\usepackage{mmap}
\usepackage[T1]{fontenc}

\usepackage[utf8]{inputenc}
\usepackage{csquotes}
\usepackage[english]{babel}

\usepackage[
    activate={true,nocompatibility},
    final, 
    tracking=true,
    kerning=true,
    spacing=true,
    factor=1100,
    stretch=10,
    shrink=10
    ]{microtype}
\microtypecontext{spacing=nonfrench}

\usepackage[backend=biber,
              style=alphabetic,
        maxbibnames=99,
      minalphanames=3,
      maxalphanames=4,
            backref=true]{biblatex}

\usepackage{hyperref}
\hypersetup{
    pdftitle={Monte Carlo to Las Vegas for Recursively Composed Functions}, 
    pdfauthor={Bandar Al-Dhalaan and Shalev Ben{-}David}, 
    colorlinks=true, 
    linkcolor=blue, 
    citecolor=blue, 
    urlcolor=green 
}

\AtEveryBibitem{
 \clearlist{address}
 \clearfield{date}
 \clearlist{location}
 \clearfield{month}
 \clearfield{series}
 \clearfield{pages}
 \clearlist{organization}
 \clearfield{number}
 \clearlist{language}

 \ifentrytype{book}{}{
  \clearlist{publisher}
  \clearname{editor}
  \clearfield{issn}
  \clearfield{isbn}
  \clearfield{volume}
 }
}

\DeclareFieldFormat{eprint:eccc}{
\mkbibacro{ECCC}\addcolon\space\ifhyperref
    {\href{http://eccc.hpi-web.de/report/#1/}{\nolinkurl{#1}}}
    {\nolinkurl{#1}}}
\DeclareFieldAlias{eprint:ECCC}{eprint:eccc}
\DeclareFieldFormat{eprint:iacr}{Cryptology ePrint\addcolon\space\ifhyperref
    {\href{https://eprint.iacr.org/#1}{\nolinkurl{#1}}}
    {\nolinkurl{#1}}}
\DeclareFieldAlias{eprint:IACR}{eprint:iacr}
\DeclareFieldFormat{eprint:arxiv}{arXiv\addcolon\space\ifhyperref
    {\href{https://arxiv.org/abs/#1}{\nolinkurl{#1}}}
    {\nolinkurl{#1}}}

\DeclareFieldFormat[article]{title}{#1}
\DeclareFieldFormat[inproceedings]{title}{#1}
\DeclareFieldFormat[incollection]{title}{#1}
\DeclareFieldFormat[online]{title}{#1}

\renewbibmacro{in:}{}

\newcommand{\lName}{1}

\newcommand{\donothing}[1]{#1}

\newcommand{\JACM}{\if\lName1\donothing{Journal of the {ACM}}\else{JACM}\fi}
\newcommand{\SICOMP}{\if\lName1\donothing{{SIAM} Journal on Computing}\else{SICOMP}\fi}
\newcommand{\ToC}{\if\lName1\donothing{Theory of Computing}\else{ToC}\fi}
\newcommand{\ToCGS}{\if\lName1\donothing{Theory of Computing Graduate Surveys}\else{ToC}\fi}
\newcommand{\TOCT}{\if\lName1\donothing{{ACM} Transactions on Computation Theory}\else{TOCT}\fi}
\newcommand{\ToIT}{\if\lName1\donothing{{IEEE} Transactions on Information Theory}\else{TOCT}\fi}
\newcommand{\CCjournal}{\if\lName1\donothing{Computational Complexity}\else{CC}\fi}
\newcommand{\CJTCS}{\if\lName1\donothing{Chicago Journal of Theoretical Computer Science}\else{CJTCS}\fi}

\newcommand{\TCS}{\if\lName1\donothing{Theoretical Computer Science}\else{TCS}\fi}
\newcommand{\IPL}{\if\lName1\donothing{Information Processing Letters}\else{IPL}\fi}
\newcommand{\JCSS}{\if\lName1\donothing{Journal of Computer and System Sciences}\else{JCSS}\fi}

\newcommand{\RSA}{\if\lName1\donothing{Random Structures and Algorithms}\else{RSA}\fi}
\newcommand{\JCTA}{\if\lName1\donothing{Journal of Combinatorial Theory, Series A}\else{JCTA}\fi}
\newcommand{\JCTB}{\if\lName1\donothing{Journal of Combinatorial Theory, Series B}\else{JCTB}\fi}
\newcommand{\PJM}{\if\lName1\donothing{Pacific Journal of Mathematics}\else{PJM}\fi}
\newcommand{\QICjournal}{\if\lName1\donothing{Quantum Information and Computation}\else{QIC}\fi}
\newcommand{\IJQI}{\if\lName1\donothing{International Journal of Quantum Information}\else{IJQI}\fi}
\newcommand{\PRA}{\if\lName1\donothing{Physical Review A}\else{PRA}\fi}
\newcommand{\PRL}{\if\lName1\donothing{Physical Review Letters}\else{PRL}\fi}
\newcommand{\VLDB}{\if\lName1\donothing{International Journal on Very Large Data Bases}\else{VLDB}\fi}

\DefineBibliographyStrings{english}{%
  backrefpage = {p{.}},
  backrefpages = {pp{.}},
}

\newtheorem{theorem}{Theorem}

\newtheorem{lemma}[theorem]{Lemma}

\newtheorem{corollary}[theorem]{Corollary}
\newtheorem{definition}[theorem]{Definition}

\newtheorem{open}{Open Problem}
\theoremstyle{definition}

\newcommand{\eq}[1]{\hyperref[eq:#1]{(\ref*{eq:#1})}}
\renewcommand{\sec}[1]{\hyperref[sec:#1]{Section~\ref*{sec:#1}}}
\newcommand{\thm}[1]{\hyperref[thm:#1]{Theorem~\ref*{thm:#1}}}
\newcommand{\lem}[1]{\hyperref[lem:#1]{Lemma~\ref*{lem:#1}}}
\newcommand{\defn}[1]{\hyperref[def:#1]{Definition~\ref*{def:#1}}}
\newcommand{\prop}[1]{\hyperref[prop:#1]{Proposition~\ref*{prop:#1}}}
\newcommand{\cor}[1]{\hyperref[cor:#1]{Corollary~\ref*{cor:#1}}}
\newcommand{\fig}[1]{\hyperref[fig:#1]{Figure~\ref*{fig:#1}}}
\newcommand{\tab}[1]{\hyperref[tab:#1]{Table~\ref*{tab:#1}}}
\newcommand{\alg}[1]{\hyperref[alg:#1]{Algorithm~\ref*{alg:#1}}}
\newcommand{\app}[1]{\hyperref[app:#1]{Appendix~\ref*{app:#1}}}
\newcommand{\conj}[1]{\hyperref[conj:#1]{Conjecture~\ref*{conj:#1}}}
\newcommand{\chap}[1]{\hyperref[chap:#1]{Chapter~\ref*{chap:#1}}}
\newcommand{\problem}[1]{\hyperref[problem:#1]{Problem~\ref{problem:#1}}}
\newcommand{\algor}[1]{\hyperref[alg:#1]{Algorithm~\ref*{alg:#1}}}

\newcommand{\tO}{\tilde{O}}

\DeclareMathOperator{\Dom}{Dom}

\newcommand{\B}{\{0,1\}}

\DeclareMathAlphabet{\mathbbold}{U}{bbold}{m}{n}

\DeclareMathOperator*{\E}{\mathbb{E}} 

\DeclareMathOperator{\bN}{\mathbb{N}}

\newcommand{\OR}{\textsc{OR}}

\newcommand{\AND}{\mathtt{AND}}
\newcommand{\MAJ}{\mathtt{MAJ}}
\newcommand{\tS}{\mathtt{S}}
\newcommand{\tI}{\mathtt{I}}
\newcommand{\PrOR}{\mathtt{PrOR}}

\newcommand{\NAND}{\mathtt{NAND}}

\newcommand{\PARITY}{\mathtt{Parity}}

\DeclareMathOperator{\D}{D}
\DeclareMathOperator{\R}{R}
\DeclareMathOperator{\Q}{Q}
\DeclareMathOperator{\C}{C}

\DeclareMathOperator{\s}{s}
\DeclareMathOperator{\bs}{bs}
\DeclareMathOperator{\fbs}{fbs}

\DeclareMathOperator{\Adv}{Adv}

\DeclareMathOperator{\M}{M} 
\DeclareMathOperator{\N}{N} 
\DeclareMathOperator{\cost}{cost}
\DeclareMathOperator{\noisyR}{noisyR}
\DeclareMathOperator{\LR}{LR}

\DeclareMathOperator{\cF}{\mathcal{F}}
\DeclareMathOperator{\cA}{\mathcal{A}}

\DeclareMathOperator{\adeg}{\widetilde{\deg}}

\begin{document}

\title{Monte Carlo to Las Vegas for Recursively Composed Functions}

\author{
Bandar Al-Dhalaan\\
\small University of Waterloo\\
\small \texttt{bandar.al-dhalaan@uwaterloo.ca}
\and
Shalev Ben{-}David\\
\small Institute for Quantum Computing\\
\small University of Waterloo\\
\small \texttt{shalev.b@uwaterloo.ca}
}

\date{}
\maketitle

\begin{abstract}
For a (possibly partial) Boolean function $f:\{0,1\}^n\to\{0,1\}$ as well as a query complexity measure $\M$ which maps Boolean functions to real numbers, define the \emph{composition limit} of $\M$ on $f$ by $\M^*(f)=\lim_{k\to\infty} \M(f^k)^{1/k}$.

We study the composition limits of general measures in query complexity.
We show this limit converges under reasonable assumptions about the measure.
We then give a surprising result regarding the composition limit of 
randomized query complexity:
we show $\R_0^*(f)=\max\{\R^*(f),\C^*(f)\}$. Among other things, this 
implies that any bounded-error randomized algorithm for recursive
3-majority can be turned into a zero-error randomized algorithm
for the same task. Our result extends also to quantum
algorithms: on recursively composed functions,
a bounded-error quantum algorithm can be converted into
a quantum algorithm that finds a certificate with high
probability.

Along the way, we prove various combinatorial properties
of measures and composition limits.
\end{abstract}

{\footnotesize
\tableofcontents}
\clearpage

\section{Introduction}

\subsection{Composed functions}

Composition is a central concept in the study
of Boolean functions: many functions of interest can
be represented as compositions of simpler functions.
For Boolean functions $f\colon\B^n\to\B$
and $g\colon\B^m\to\B$, their composition
(sometimes called ``block composition'') is the Boolean
function $f\circ g\colon\B^{nm}\to\B$ defined
by applying $g$ to $n$ independent inputs, and plugging
the resulting $n$-bit string into $f$ to get a $1$-bit
answer. (This definition can be extended to partial
functions, which are defined on a subset of $\B^n$;
see \defn{composition}.)

An important line of work in
query complexity tries to establish
\emph{composition theorems} for certain
measures of Boolean functions;
such theorems aim to relate the complexity
of the composed function to the complexities of the
original function, particularly in the lower bound
direction. For example, a well-known composition
theorem for quantum query complexity $\Q(f)$ follows from
the ``negative-weight adversary bound''
\cite{HLS07,Rei11,LMR+11,Kim13}, and shows that
$\Q(f\circ g)=\Theta(\Q(f)\Q(g))$.

In this work, we focus on \emph{recursively composed}
functions, sometimes called tree functions.
That is, using $f^k$ to denote the composition
$f\circ f\circ \dots\circ f$ (with $k$ levels of $f$),
we wish to understand the limiting behavior of $f^k$
as $k\to\infty$. Recursively composing a Boolean function
has become a standard tool in query complexity: functions
constructed this way are commonly used to provide
separations between query complexity measures,
among many other uses. For some examples of specific
recursively composed functions, see
\cite{SW86,WZ88,NW95,KRW95,Aar08,Tal13,Goo15,KRS15,Amb16,GSS16,
ABK16,JK17,BKT18,BBG+21}.

Several works have also studied recursive composition
(or related notions)
in the abstract, including \cite{San95,Tal13,GJ16,GSS16,EMP18}.
Such works often study (explicitly
or implicitly) the \emph{composition limit} of a measure
of Boolean functions. That is, if $\M(f)$ is a real-valued
measure of a Boolean function
(such as certificate complexity or randomized query complexity),
define its composition limit by
\[\M^*(f)\coloneqq \lim_{k\to\infty} \M(f^k)^{1/k}.\]
Since $\M(f^k)$ will generally increase exponentially in $k$
for most measures $\M$ of interest, this limit
(when it exists) should itself be some nontrivial measure of $f$.
For some values of $\M$, such as quantum query complexity $\Q(f)$
and deterministic query complexity $\D(f)$,
the behavior of $\M^*(f)$ is well-understood
(we have $\Q^*(f)=\Adv^{\pm}(f)$ \cite{HLS07,Rei11,LMR+11,Kim13} 
and $\D^*(f)=\D(f)$ \cite{Tal13,Mon14}).
However, for other measures $\M$, the behavior of $\M^*$
seems to be extremely complex.

\subsection{Randomized query complexity}

We are motivated in part by the study of randomized
query complexity $\R(f)$ for composed functions.
Even for just two functions $f$ and $g$, the randomized
query complexity of the composition, $\R(f\circ g)$,
has complex behavior, and has been the subject of a lot
of work 
\cite{AGJ+18,BK18,GJPW18,GLSS19,BDG+20,BB20b,BGKW20,GM21,BBGM22,San24}.
For recursively composed functions, our understanding
is significantly worse.
In fact, even the following question is open:

\begin{open}\label{problem:one}
Is $\R^*(f)$ computable? That is, given the truth
table of a function $f$, can the value of $\R^*(f)$
be computed (to a specified precision) on a Turing machine?
\end{open}

The measure $\R^*(f)$ has mostly been studied for two
specific choices of $f$. The first is $\NAND_2$,
the NAND function on $2$ bits. In this case, $f^k$
becomes the ``NAND-tree'' function, also known as the
AND-OR-tree or game tree. This family of functions $f^k$
is known to separate $\R(\cdot)$ from $\D(\cdot)$
(in fact, it was the best known separation for a total function
before \cite{ABB+17}), because the exact value of $\R^*(\NAND_2)$
is known to be $(1+\sqrt{33})/4\approx 1.686$ \cite{SW86,San95},
so that $\R(\NAND_2^k)\approx 1.686^k\approx (2^k)^{\log_2(1.686)}$
while $\D(\NAND_2^k)= 2^k$. It is also known that
$\R_0^*(\NAND_2)=\R^*(\NAND_2)$, where $\R_0(f)$ is the zero-error
(``Las Vegas'', ``$\mathsf{ZPP}$'')
randomized query complexity of $f$. This originally
led to a conjecture that $\R_0(f)=\Theta(\R(f))$ for all
total Boolean functions, though this conjecture was disproved
in \cite{ABB+17}.

The second choice of $f$ for which $\R^*(f)$ has received significant attention
is $f=\MAJ_3$, the majority function on $3$ bits.
Despite much work \cite{JKS03,MNS+15,Leo13,GJ16},
the value of $\R^*(\MAJ_3)$
is still not known! It is only known to be between
$2.596$ and $2.650$. The value of $\R_0^*(\MAJ_3)$ is also open.
We note that if an algorithm as in \problem{one} were known,
one could use it to simply compute $\R^*(\MAJ_3)$,
since the truth table of $\MAJ_3$ has only $8$ bits of input.

\subsection{Our results}

Our main result is as follows.

\begin{theorem}\label{thm:R0RC}
For all (possibly partial) Boolean functions $f$,
$\R_0^*(f)=\max\{\R^*(f),\C^*(f)\}$.
\end{theorem}
Here $\R_0$ is the zero-error randomized query complexity,
$\R$ is the randomized query complexity, and $\C$
is the certificate complexity (with the $*$ denoting the
composition limit).
In other words, we show that for recursively
composed functions, a ``Monte Carlo''
(a.k.a. bounded error, $\mathsf{BPP}$)
randomized algorithm can be converted into a ``Las Vegas''
(a.k.a. zero error, $\mathsf{ZPP}$)
randomized algorithm with no additional cost
unless the certificate complexity is large.
Note that since
a zero-error randomized query algorithm must always
find a certificate, it is easy to see that
$\R_0^*(f)\ge \C^*(f)$, and since bounded-error algorithms
can simulate zero-error algorithms,
$\R_0^*(f)\ge \R^*(f)$. This means that the lower bound
direction of \thm{R0RC} is easy, and our main
contribution is the upper bound on $\R_0^*(f)$.

\thm{R0RC} implies, in particular, that
$\R_0^*(\MAJ_3)=\R^*(\MAJ_3)$. Even this special case
was not previously known, despite significant work
on these two measures \cite{JKS03,MNS+15,Leo13,GJ16}.

More generally, \thm{R0RC} gives a strategy which
uses only a bounded-error algorithm to find a certificate
in a recursively composed function.
This strategy
also works to convert a bounded error quantum algorithm
into a quantum algorithm for finding a certificate.

\begin{theorem}\label{thm:Q0QC}
Let $\Q_{\C}(f)$ denote the number of quantum queries
required to find a certificate of $f$ with constant
success probability. Then
\[\Q_{\C}^*(f)=\max\{\Q^*(f),\C^*(f)\}.\]
\end{theorem}
The precise definition of $\Q_{\C}(f)$ is given
in \sec{LasVegas}. We note that our approach is likely to
also work for other computational models:
for example, it is likely possible to construct
a certificate-finding system of polynomials for a
recursively composed function whose degree is only
the approximate degree; we do not do this here, but
the only likely barriers are definitional technicalities.

On the way to \thm{R0RC} and \thm{Q0QC}, we encounter
a perhaps surprising issue: proving that composition
limits such as $\R^*(f)$, $\R_0^*(f)$, and $\Q_{\C}^*(f)$
converge is surprisingly tricky. To address this,
we prove some structural results regarding composition limits
for general measures.

\begin{theorem}[Informal; see \thm{FormalLimit}]\label{thm:limit}
Let $\M$ be a measure of Boolean functions
such that $\M(f\circ g)= \tO(\M(f)\M(g))$
and $\M(f\circ g)=\Omega(\M(f))$
holds for all $f$ and $g$.
Under some minor nicety conditions on $\M$,
the limit $\M^*(f)=\lim_{k\to\infty}\M(f^k)^{1/k}$
converges.
\end{theorem}

The conditions in this theorem are relatively easy to
satisfy. We therefore get the following corollary.

\begin{corollary}\label{cor:converge}
The following measures all have convergent composition limits:
\begin{enumerate}
\item The deterministic, randomized, and quantum query complexities
$\D(f)$, $\R(f)$, $\Q(f)$, and their zero-error and exact variants
$\R_0(f)$, $\Q_{\C}(f)$, and $\Q_E(f)$
\item Certain ``local'' measures such as sensitivity $\s(f)$,
fractional block sensitivity $\fbs(f)$,
and certificate complexity $\C(f)$ (also follows from \cite{GSS16})
\item Polynomial degree measures such as degree $\deg(f)$ and
approximate degree $\adeg(f)$, even for partial functions $f$.
\end{enumerate}
\end{corollary}

\paragraph{Convenience theorems.}
We also establish other structural results regarding
composition limits. For example, we show that relations
between measures such as $\M_1(f)\le \M_2(f)$
imply the corresponding relations for their composition 
limits, and we show that lower-order terms often
disappear when taking composition limits. We also
establish that $(\M^*)^*(f)=\M^*(f)$ (that is, the composition
limit of the composition limit is the original composition
limit), as well as other results such as
$\M^*(f\circ g)=\M^*(g\circ f)$ and $\M^*(f^k)=\M^*(f)^k$.

\paragraph{Combinatorial reductions.}
In order to show that measures satisfy the nicety conditions
of \thm{limit}, we study combinatorial properties of 
measures. We say a measure is well-behaved if it is
invariant under renaming of the indices (permuting
all input strings in the same way), invariant under
duplication of bits (adding a new bit to each input
which, in the promise of the function, always
takes an identical value to another existing bit),
invariant under superfluous bits (adding a bit to each
input that is not used for determining the function value),
and non-increasing under restrictions to a promise.
Another useful invariance property is what we call
alphabet renaming: negating a bit in all inputs
to the function. A well-behaved measure invariant under
alphabet renaming is called strongly well-behaved.

We define a notion of reductions between functions
which amounts to applying the aforementioned transformations.
We prove a variety of properties of how reductions
interact with compositions. A highlight of our results
is the definition of the switch function: this is 
the function $\tS\colon\{01,10\}\to\B$ defined by $\tS(01)=0$
and $\tS(10)=1$. We show that the switch function is the
easiest non-constant function for well-behaved measures,
in the following sense.

\begin{theorem}
For all (possibly partial)
non-constant functions $f$ and $g$
and all well-behaved measures $\M(\cdot)$,
we have $\M(f\circ g)\ge \M(f\circ \tS)$ and
$\M(f\circ g)\ge \M(\tS\circ g)$.
\end{theorem}

We can also show that the measure of a function
can be lower bounded by a block sensitivity lower bound:
the measure applied to the promise-OR function $\PrOR$
(composed with switch), of size equal to the block
sensitivity $\bs(f)$.

\begin{theorem}
For any (possibly partial) Boolean function $f$
and any well-behaved measure $\M(\cdot)$ on
Boolean functions,
$\M(f)\ge \M(\PrOR_{\bs(f)}\circ \tS)$.
\end{theorem}

\subsection{Our techniques}

\paragraph{Monte Carlo to Las Vegas.}
To establish \thm{R0RC}, we need to show how one can
use a bounded-error randomized algorithm $R$ for $f^k$
in order to find a certificate for $f^k$ (assuming
$k$ is large compared to the size of $f$).
The function $f^k$ can be represented as a tree of depth
$k$, where each node has $n$ children and applies
the gate $f$ to their values; at the bottom
of the tree is the input of length $n^k$.

Our approach to constructing a zero-error algorithm is
to use a bounded-error algorithm
to evaluate each of the $n$ children of the root of the tree;
that is, we run a bounded-error algorithm for
each of the $n$ copies of $f^{k-1}$ that are fed into
the outermost copy of $f$. We amplify each of these runs
to reduce the error. This gives us a string of length $n$
of the best-guess values for the input to the outermost
$f$, and with high probability (but not certainty),
this guess string is correct.

Next, assuming this $n$-bit string is the true input,
we choose a certificate $c$ for it which is cheapest.
Finally, we recursively call the
zero-error algorithm on each of the copies of $f^{k-1}$
that are used by the certificate $c$; this algorithm
produces a certificate for each of these copies,
and together they form a certificate for the outer copy
of $f$, which we return.

An analysis of this protocol eventually gives a bound
of the form
\[\R_0(f^k)\le \tO(n)\cdot \tO(\R(f)+\C(f))^k.\]
Taking a power of $1/k$ on both sides and the limit
as $k\to\infty$, the $\tO(n)$ factor disappears, and
we are left with $\R_0^*(f)\le \tO(\R(f)+\C(f))$.
This is still too large of a bound: $\R^*(f)$ and
especially $\C^*(f)$ may be substantially smaller
than $\R(f)$ and $\C(f)$. The next trick is perhaps
counterintuitive: we just take the composition limit
of the inequality itself. In other words, we apply
the $*$ operator to both sides:
\[\R_0^{**}(f)\le \tO(\R(f)+\C(f))^*.\]
Using our convenience theorems, we show this simplifies
drastically: $\R_0^{**}(f)$ equals $\R_0^*(f)$,
log factors disappear, the addition becomes a max,
and the measures $\R(f)$ and $\C(f)$ become the potentially
smaller measures $\R^*(f)$ and $\C^*(f)$.
We can therefore extract
a clean and powerful upper bound out of the much cruder
analysis of our simple Las Vegas algorithm:
\[\R_0^*(f)\le \max\{\R^*(f),\C^*(f)\}.\]

\paragraph{Convergence of composition limits.}
For the convergence of composition limits 
(\thm{limit}), roughly speaking,
we form a subsequence that converges to both
the limsup and the liminf of the sequence $\M(f^k)^{1/k}$.
The idea is as follows: we interleave a sequence 
$k_1,k_3,k_5\dots$ which converges to the limsup of
$\M(f^k)^{1/k}$ with a sequence $k_2,k_4,k_6,\dots$
which converges to the liminf of $\M(f^k)^{1/k}$.
However, we do this carefully, picking each $k_i$ one
at a time, and we pick each one to be so large compared
to the previous ones that it is close
(in a multiplicative sense)
to a multiple of all the previous ones. We can therefore
round the $k_i$ in our sequence so that each one is
exactly equal to a multiple of all the previous ones,
and we do this in a careful way that ensures the odd
indices still converge to the limsup and the even ones
to the liminf.

Finally, we use the composition properties that we assumed
about $\M$ to upper bound the limsup subsequence in terms
of the liminf subsequence, forcing them to approach each
other and converge to the same value. This suffices
to show the limit of the original sequence exists.

\section{Preliminaries}\label{sec:prelims}

\subsection{Boolean functions, measures, and composition limits}

We review some basic definitions in query complexity.
We warn that some of the details,
such as the definition of a general
measure on Boolean functions, are not necessarily standard
in the field, since such definitions are not usually needed
in this generality.

\begin{definition}[Boolean functions]
A possibly partial Boolean function on $n\in\bN$ bits is a function
$f\colon S\to\B$, where $S\subseteq\B^n$. We can also write $f$
as a function $f\colon\B^n\to\{0,1,*\}$, with $f(x)=*$ for $x\notin S$.
The domain of $f$ is $\Dom(f)\coloneqq S$, and the input size
is $n(f)\coloneqq n$. The function $f$ is called \emph{total}
if $\Dom(f)=\B^{n(f)}$.
\end{definition}

\begin{definition}
For $n\in\bN$, let $\mathcal{F}_n$ denote the set of all possibly
partial Boolean functions on $n$ bits. Let
$\mathcal{F}\coloneqq\bigcup_{n=1}^\infty\mathcal{F}_n$
denote the set of all (possibly partial) Boolean functions.
\end{definition}

\begin{definition}[Measure]
A measure $M$ on a subset $\mathcal{A}\subseteq\mathcal{F}$
is a function $M\colon\mathcal{A}\to[0,\infty)$.
\end{definition}

\begin{definition}[Composition]\label{def:composition}
Composition is a binary operation 
$\circ\colon\mathcal{F}\times\mathcal{F}\to\mathcal{F}$
on possibly partial Boolean functions. For two functions $f,g\in\mathcal{F}$,
their composition $f\circ g$ is a function on strings of length $n(f)\cdot n(g)$
representing the composition of $f$
with $n(f)$ independent copies of the function $g$.

Formally, a string $x$ of length $n(f)\cdot n(g)$ is in $\Dom(f\circ g)$ if
it can be written $x=y^1y^2\dots y^{n(f)}$ with each $y^i\in\Dom(g)$,
and if in addition the string $g(y^1)g(y^2)\dots g(y^{n(f)})$ is in $\Dom(f)$.
In that case, the value of $f\circ g(x)$ is defined to be
$f(g(y^1)g(y^2)\dots g(y^{n(f)}))$.
\end{definition}

\paragraph{Remark.} It is not hard to see that composition
is associative. Denote $f\circ f$ by $f^2$,
$f\circ f\circ f$ by $f^3$, and so on.
Using $I$ to denote the identity function on one bit,
we observe that $(\mathcal{F},\circ)$ forms a \emph{monoid}
with identity $I$, since $f\circ I=I\circ f=f$ for all
$f\in\mathcal{F}$. We define $f^0$ to be $I$.

\begin{definition}[Composition-closed class]
A subset $\mathcal{A}\subseteq\mathcal{F}$ is called
a \emph{composition-closed class} if it is a submonoid
of $\mathcal{F}$; that is, it must contain $I$ and be
closed under composition.
\end{definition}

\begin{definition}[Composition limits]
For any $f\in\mathcal{F}$ and any measure $M$ defined on $\{f,f^2,\dots\}$,
define
\[M^{\underline{*}}(f)\coloneqq\liminf_{k\to\infty} M(f^k)^{1/k},\]
\[M^{\overline{*}}(f)\coloneqq\limsup_{k\to\infty} M(f^k)^{1/k}.\]
If $M^{\underline{*}}(f)=M^{\overline{*}}(f)$, denote
this quantity by $M^*(f)$; this is called the composition limit of $M$
applied to $f$, and is equal to $\lim_{k\to\infty} M(f^k)^{1/k}$.
\end{definition}

\subsection{Decision trees, certificates, and randomized algorithms}

\begin{definition}[Decision tree]
A decision tree on $n$ bits is a rooted binary tree
with internal nodes labeled by $[n]$, leaves labeled by
$\B$, and arcs labeled by $\B$. We additionally require
that no pair of internal nodes such that one is an ancestor
of the other share the same label; this ensures
the height each leaf is at most $n$. The height of the
tree is defined as the maximum height from the root to a leaf.

For a decision tree $D$ on $n$ bits, let $D(x)$
denote the leaf label reached when we start from the root,
and for each internal node labeled by $i$, we follow
the arc labeled by $x_i$ down the tree. For a possibly
partial Boolean function $f\in\mathcal{F}$, we say
that $D$ computes $f$ if $n=n(f)$ and $D(x)=f(x)$
for all $x\in\Dom(f)$. The deterministic query complexity,
$\D(f)$, is defined as the minimum height of a decision
tree which computes $f$.
\end{definition}

\begin{definition}[Partial assignment]
A partial assignment on $n$ bits is a string
$p\in\{0,1,*\}^n$. We equate $p$ with both the set of
pairs $\{(i,p_i):p_i\ne *\}$ as well as with the function
defined by this set of pairs. We say partial assignments
$p$ and $q$ are consistent if for all $i\in[n]$,
either $p_i=q_i$ or else at least one of $p_i$ and $q_i$
is $*$. The use of set notation such as $p\subseteq q$
and $|p|$ should be interpreted with respect to the
set of ordered pairs.
\end{definition}

\begin{definition}[Certificates]\label{def:certificate}
Let $f\in\mathcal{F}$ be a possibly partial
Boolean function. Let $x\in\Dom(f)$, and let
$p$ be a partial assignment on $n(f)$ bits.
We say that $p$ is a certificate for $x$ if
$p\subseteq x$ and
for all $y\in\Dom(f)$ with $p\subseteq y$,
we have $f(y)=f(x)$. The certificate complexity
of $f$ at $x$, denoted $\C_x(f)$, is the minimum
size $|p|$ of a certificate of $x$ with respect to $f$.

The certificate complexity of $f$, denoted $\C(f)$,
is defined as $\C(f)\coloneqq\max_{x\in\Dom(f)}\C_x(f)$.
We also define $\C_0(f)\coloneqq\max_{x\in f^{-1}(0)} \C_x(f)$
and $\C_1(f)\coloneqq\max_{x\in f^{-1}(1)}\C_x(f)$.
\end{definition}

\begin{definition}[Randomized query complexity]
A randomized query algorithm $R$ on $n$ bits is a probability
distribution over decision trees on $n$ bits.
On each string $x\in\B^n$, we define $\cost(R,x)$
to be $\E_{D\sim R}[\cost(D,x)]$, where $\cost(D,x)$
for a decision tree $D$ is the length of the path
from the root to the leaf reached by $x$;
in other words, $\cost(R,x)$ is the expected number
of queries $R$ makes when run on $x$. The height
of $R$ is defined as the maximum height of any decision
tree in the support of $R$.

We define $R(x)$ to be the random variable $D(x)$
when $D$ is sampled from $R$. We say that $R$
computes $f$ to error $\epsilon\in[0,1/2)$ if
for all $x\in\Dom(f)$, we have
$\Pr[R(x)= f(x)]\ge 1-\epsilon$.
When $\epsilon>0$,
we define $\R_\epsilon(f)$ to be the minimum
height of a randomized query algorithm $R$ which computes
$f$ to error $\epsilon$. We further define
$\overline{\R}_\epsilon(f)$ to be the minimum
of $\max_{x\in\Dom(f)}\cost(R,x)$ over randomized
query algorithms $R$ which compute $f$ to error $\epsilon$.
When $\epsilon=1/3$, we omit it and write
$\R(f)$ and $\overline{\R}(f)$.

When $\epsilon=0$, the measure $\R_\epsilon(f)$
becomes the same as that of $\D(f)$, but the measure
$\overline{\R}_\epsilon(f)$ stays distinct. Following
convention in the literature, we define
$\R_0(f)$ to be $\overline{\R}_0(f)$.
\end{definition}

The definitions above are all standard in query complexity.
We note that $\R_\epsilon(f)$ and
$\overline{\R}_\epsilon(f)$ can both be amplified
(repeating the algorithm a few times to reduce the error
while increasing the cost), which means that the value
of $\epsilon$ does not matter if it is a constant in $(0,1/2)$
and if we do not care about multiplicative constants. 
Moreover, these measures are non-increasing in $\epsilon$,
and $\R_\epsilon(f)\ge\overline{\R}_\epsilon(f)$.
Markov's inequality can be used to cut off an algorithm
that is running too long compared to its expectation;
this can be used to show that $\overline{R}(f)=O(\R(f))$,
so the two measures are equivalent up to constant
factors. We also have
$\R_0(f)\ge \overline{\R}(f)=\Omega(\R(f))$.

In summary, when $\epsilon$ is either $0$ or constant,
the only distinct measures are $\D(f)$, $\R_0(f)$,
and $\R(f)$, each of which is smaller than the last
(up to constant factors). They correspond to the
complexity classes $\mathsf{P}$, $\mathsf{ZPP}$,
and $\mathsf{BPP}$, respectively.

\subsection{Sensitivity and degree measures}

\begin{definition}[Block notation]
A set of indices $B\subseteq[n]$ is called a \emph{block}.
Given a string $x\in\B^n$ and a block $B\subseteq[n]$,
the string $x^B$ is defined as the string $x$ with the bits
in $B$ flipped, i.e. $x^B_i=x_i$ if $i\notin B$ and
$x^B_i=1-x_i$ if $i\in B$. If $B=\{i\}$ contains a single
bit, we use $x^i$ as shorthand for $x^{\{i\}}$.
\end{definition}

\begin{definition}[Sensitivity]
Let $f\in\cF$ and let $x\in\Dom(f)$.
A bit $i\in[n]$ is called \emph{sensitive} for $x$
with respect to $f$ if $x^i\in\Dom(f)$ and $f(x^i)\ne f(x)$.
The \emph{sensitivity} of $x$ with respect to $f$, denoted $\s_x(f)$,
is the number of bits $i\in[n]$ which are sensitive for $x$.
The sensitivity of $f$ is defined as 
$\s(f)\coloneqq\max_{x\in\Dom(f)}\s_x(f)$.
\end{definition}

\begin{definition}[Block sensitivity]
Let $f\in\cF$ and let $x\in\Dom(f)$.
A block $B\subseteq[n]$ is called \emph{sensitive}
for $x$ with respect to $f$ if $x^B\in\Dom(f)$
and $f(x^B)\ne f(x)$. The \emph{block sensitivity} of $x$ with
respect to $f$, denoted $\bs_x(f)$, is the maximum number $k\in\bN$
such that there are $k$ blocks $B_1,\dots,B_k\subseteq[n]$
which are pairwise disjoint and which are all sensitive for $x$.
The block sensitivity of $f$ is defined as
$\bs(f)\coloneqq\max_{x\in\Dom(f)}\bs_x(f)$.
\end{definition}

\begin{definition}[Fractional block sensitivity]
Let $f\in\cF$ and let $x\in\Dom(f)$.
Let $\mathcal{B}$ be the set of sensitive blocks of $x$
with respect to $f$.
The \emph{fractional block sensitivity} of $x$ with respect
to $f$, denoted $\fbs_x(f)$, is the maximum possible sum
$\sum_{B\in\mathcal{B}} w_B$, where the weights $w_B\ge 0$
are constrained to satisfy $\sum_{B\in\mathcal{B}:i\in B} w_B\le 1$
for all $i\in[n]$.
The fractional block sensitivity of $f$ is defined as
$\fbs(f)\coloneqq\max_{x\in\Dom(f)}\fbs_x(f)$.
\end{definition}

The definition of other measures, such as polynomial degree and exact
quantum query complexity, can be found in \cite{BdW02}.

\subsection{Subtleties of little-o notation}

Since little-$o$ notation will often occur in the exponents
of our query complexity measures, we take a moment
to clarify exactly what we mean by this notation. The
details can get a bit subtle.

\begin{definition}[Little-o notation]\label{def:littleo}
Let $S$ be an infinite set and let $N,M\colon S\to[0,\infty)$ be functions.
We say that $M(x)\le N(x)^{o(1)}$ (over $x\in S$)
if $M(x)=0$ whenever $N(x)=0$ and if for every $\epsilon>0$,
there exists $C>0$ such that $M(x)\le N(x)^{\epsilon}$ whenever
$M(x)\ge C$.

We extend this definition in the natural way; for example,
for measures $M,N_1,N_2\colon\mathcal{F}\to[0,\infty)$,
the statement $M(f\circ g)\le N_1(f)^{1+o(1)}N_2(g)^{1+o(1)}$
means the same thing as
$\frac{M(f\circ g)}{N_1(f)N_2(g)}\le (N_1(f)N_2(g))^{o(1)}$
(plus the condition that $M(f\circ g)=0$ when the denominator
$N_1(f)N_2(g)$ is $0$),
where both sides are now functions from $\mathcal{F}^2$ to $[0,\infty)$
(so that the previous definition applies).
Similarly, $M(f\circ g)\ge N_1(f)^{1-o(1)}N_2(g)^{1-o(1)}$
means the same thing as
$\frac{N_1(f)N_2(g)}{M(f\circ g)}\le (N_1(f)N_2(g))^{o(1)}$
(plus the condition that $N_1(f)N_2(f)=0$ whenever $M(f\circ g)=0$).
\end{definition}

\begin{lemma}
The statement $M(x)\le N(x)^{1+o(1)}$ as defined in
\defn{littleo} is equivalent to the statement $N(x)\ge M(x)^{1-o(1)}$.
\end{lemma}

\begin{proof}
Note that these two definitions are somewhat different. The former says
that for each $\epsilon>0$ there exists $C_\epsilon$ such that $M(x)/N(x)\le N(x)^\epsilon$
whenever $M(x)/N(x)> C_\epsilon$, or equivalently, that for each $\epsilon>0$
there exists $C_\epsilon$ so that for all $x$,
$M(x)\le \max\{N(x)^{1+\epsilon},C_\epsilon N(x)\}$.
The latter says that for each $\epsilon>0$ there exists $C'_\epsilon$
such that $M(x)/N(x)\le M(x)^\epsilon$ whenever $M(x)/N(x)> C'_\epsilon$,
or equivalently, that for each $\epsilon>0$ there exists
$C'_\epsilon$ so that for all $x$,
$M(x)\le \max\{N(x)M(x)^{\epsilon},C'_\epsilon N(x)\}$.
For $\epsilon<1$, we can rearrange this latter condition as
$M(x)\le \max\{N(x)^{1/(1-\epsilon)},C'_\epsilon N(x)\}$.
(Both statements also include the condition that $N(x)=0\Rightarrow M(x)=0$,
but since this condition is present in both we may assume $N(x)\ne 0$
for this proof.)

The only difference between the formal version of the two conditions
is therefore the exponent on $N(x)$: it is either $1+\epsilon$ or $1/(1-\epsilon)$.
For $\epsilon<1$, we can convert between the two just by using a different
value of $\epsilon$, so the two conditions are equivalent.
\end{proof}

\section{Combinatorial properties of measures}

\subsection{Basic properties}

Since we aim to study measures on Boolean functions
with a high degree of generality, we will start by
defining some basic conditions we expect such measures
to satisfy. The first such condition, called
index renaming, says that a measure $\M$
(such as $\R(f)$ or $\C(f)$) should be invariant
to renaming the indices of the input string
(that is, permuting the bits of all the inputs to a function).

\begin{definition}[Index renaming]
Let $f\in\mathcal{F}_n$ be a (possibly partial) Boolean
function on $n$ bits, and let $\pi\colon[n]\to[n]$ be a permutation.
For any $x\in\B^n$, let $x_{\pi}\in\B^n$ denote the shuffling
of the bits of $x$ according to $\pi$, i.e.
$(x_{\pi})_i\coloneqq x_{\pi(i)}$. The \emph{index renaming}
of $f$ according to $\pi$ is the function $f_\pi$
defined on domain $\{x:x_\pi\in\Dom(f)\}$ via
$f_{\pi}(x)\coloneqq f(x_{\pi})$.

We say that a measure $\M$ defined on
$\mathcal{A}\subseteq\mathcal{F}$ is \emph{invariant under
index renaming} if for all $f\in \cA$ and all permutations
$\pi$, we have $f_{\pi}\in\mathcal{A}$ and $\M(f_\pi)=\M(f)$.
\end{definition}

The next property says that adding additional bits to the input
that don't affect the output of the function should
not change $\M(f)$.

\begin{definition}[Superfluous bits]
Let $f\in\mathcal{F}$ be a (possibly partial) Boolean function
on $n$ bits, and let $S\subseteq\B^m$ for some $m\in\bN$.
Define the \emph{superfluous bits} modification to $f$
with respect to $S$ as the function $f_S$ defined on
$\{xy:x\in\Dom(f),y\in S\}$ via $f_S(xy)\coloneqq f(x)$.

We say that a measure $\M$ defined on $\cA\subseteq\cF$
is \emph{invariant under superfluous bits} if for all
$f\in\cF$ and $S\subseteq\B^m$, we have $f\in \cA$
if and only if $f_S\in\cA$, and if $f\in\cA$ then
$\M(f_S)=\M(f)$.
\end{definition}

The above two properties are satisfied by virtually
all query measure of interest (up to some technical
details such some measures being defined only
on total functions, and hence not technically satisfying
the superfluous bits condition).
The next property says that a measure $\M$
should be invariant to duplicating bits.
This one is satisfied by most query measures; of the
ones we defined in \sec{prelims},
the only one that fails it is sensitivity $\s(f)$,
since sensitivity cares about the individual bits,
so turning a bit into a block of identical bits
changes the sensitivity.

\begin{definition}[Bit duplication]
Let $f\in\mathcal{F}$ be a (possibly partial) Boolean function
on $n$ bits, and let $i\in[n]$. Define the
\emph{bit duplication} of $f$ at bit $i$ to be the function
$f_i$ on domain $\{xx_i:x\in\Dom(f)\}$ defined by
$f_i(xx_i)=f(x)$ for all $x\in\Dom(f)$. Here $xx_i$
is the concatenation of $x$ with the additional bit $x_i$.

We say that a measure $\M$ defined on $\cA\subseteq\cF$
is \emph{invariant under bit duplication} if for all
$f\in\cF$ and all $i\in[n]\cup\{0\}$, we have
$f\in\cA$ if and only if $f_i\in\cA$ and
if $f\in\cA$ then $\M(f_i)=\M(f)$.
\end{definition}

Next we define the ``alphabet renaming''
property, which refers to renaming the input alphabet
from ``0'' and ``1'' to ``1'' and ``0'' respectively.
We allow this renaming to happen for only some subset
of the $n$ positions of the string. Nearly
all measures of interest are invariant under such renaming,
but unfortunately, measures defined by composition
(such as the composition limit $\M^*(f)$) are not
invariant under this property.

\begin{definition}[Alphabet renaming]
Let $f\in\mathcal{F}_n$ be a (possibly partial) Boolean
function on $n$ bits, and let $z\in\B^n$ be a string.
Define the \emph{alphabet renaming} of $f$ according to $z$
to be the function $f_z$ on domain $\{x:x\oplus z\in\Dom(f)\}$
defined by $f_z(x)\coloneqq f(x\oplus z)$, where $\oplus$
denotes the bitwise XOR of the two strings.

We say that a measure $\M$ defined on $\cA\subseteq \cF$
is \emph{invariant under alphabet renaming} if for all
$f\in\cA$ and all $z\in\B^n$, we have $f_z\in\cA$ and
$\M(f_z)=\M(f)$.
\end{definition}

Some measures are ``two-sided'', which means
they do not care if the output of $f$ is negated.
Other ``one-sided'' measures such as $\C_1(f)$ do care
about this. Measures defined my composition are generally
not two-sided.

\begin{definition}[Two sided]
Let $f\in\mathcal{F}_n$ be a (possibly partial) Boolean
function on $n$ bits. Define the \emph{negation} of
$f$, denoted $\overline{f}$, to be the function
defined on $\Dom(f)$ via $\overline{f}(x)\coloneqq 1-f(x)$.

We say that a measure $\M$ defined on $\cA\subseteq\cF$
is \emph{two-sided} if it is invariant under negations:
that is, for all $f\in\cA$, we have $\overline{f}\in\cA$
and $\M(\overline{f})=\M(f)$.
\end{definition}

Finally, virtually all measures respect the promise
in the sense that restricting to a smaller promise
cannot increase the measure.

\begin{definition}[Promise respecting]
We say a measure $\M$ defined on $\cA\subseteq \cF$
is \emph{promise respecting} if for any $f\in\cA$
and any $P\subseteq\Dom(f)$, the restriction $f|_P$
of $f$ to the subdomain $P$ satisfies $f|_P\in \cA$
and $\M(f|_P)\le \M(f)$.
\end{definition}

We collect the above definition into one definition
for convenience.

\begin{definition}[Well-behaved]
A measure $\M$ defined on $\cA\subseteq\cF$
is called \emph{weakly well-behaved} if it is invariant
under index renaming, superfluous bits, and bit duplication,
and in addition, it is promise-respecting.

We say $\M$ is \emph{strongly well-behaved} (or just
well-behaved) if in addition it is invariant under
alphabet renaming.
\end{definition}

\begin{definition}[Reductions]
For (possibly partial) Boolean functions $f$ and $g$,
we write $f\lesssim g$ if we can convert from $g$ to $f$
using the operations of index renaming, superfluous bits,
and bit duplication (or their inverse operations,
such as removing duplicated bits) as well as restriction
to a promise. We write $f\lesssim' g$ if we can convert
$g$ to $f$ using these operations in combination with
alphabet renaming (negating input bits).
\end{definition}

We note that $\lesssim$ and $\lesssim'$ are transitive
relations. They also characterize whether a measure
is well-behaved.

\begin{lemma}\label{lem:WellBehavedReductions}
A measure $\M$ is weakly well behaved if and only if
$f\lesssim g\Rightarrow \M(f)\le \M(g)$ for all $f$ and $g$.
Similarly, $\M$ is strongly well behaved if and only if
$f\lesssim' g\Rightarrow \M(f)\le\M(g)$ for all $f$ and $g$.
\end{lemma}

\begin{proof}
It follows immediately from definitions that for a well-behaved
measure respects the corresponding reductions.
For the converse direction,
we want to show that $\M$ is well-behaved assuming it
respects reductions. Being well-behaved means being non-increasing
under restriction to a promise, and invariant under
the other operations. The former property follows immediately
from the fact that $\M$ respects reductions (since restriction
to a promise is a type of reduction). For the latter,
consider the bit duplication property. We need to show that
if $f'$ is a bit-duplication of $f$, then $\M(f)=\M(f')$.
However, if $f'$ is a bit-duplication of $f$, then
$f\lesssim f'$ and $f'\lesssim f$, so we know that
$\M(f)\le\M(f')$ and $\M(f')\le \M(f)$, so $\M(f)=\M(f')$
as desired. Since all the other properties in the
definitions of well-behaved are reversible in this way
(except restriction to a promise), the same argument
works for all of them.
\end{proof}

\subsection{Switchable functions and reductions for composed functions}

\begin{definition}[Switchable function]
A function $f$
is called \emph{switchable} if $\overline{f}\lesssim' f$.
It is called \emph{strongly switchable} if
$\overline{f}\lesssim f$.

We also define the \emph{switch function}
to be $\mathtt{S}\colon\{01,10\}\to\B$
defined by $\mathtt{S}(01)=0$ and $\mathtt{S}(10)=1$.
\end{definition}

\paragraph{Remark.} Many familiar functions, such as $\MAJ$
and $\PARITY$, are switchable (for majority, negating
all its bits negate the output; for parity, negating a single
bit negates the output). Fewer natural functions are strongly switchable,
though we will see in \lem{Sof} that composition with
the switch function gives rise to them. Functions like
$\AND$ and $\OR$ do not seem to be switchable (we do not prove this
here).

\begin{lemma}\label{lem:NegationSwitch}
$f$ is (strongly) switchable if and only if $\overline{f}$ is 
(strongly) switchable.
\end{lemma}

\begin{proof}
Suppose $f$ is switchable. Then $\overline{f}\lesssim' f$.
Consider applying this reduction to $\overline{f}$
instead of to $f$. Note that the operations in the reduction
(such as adding/removing duplicate bits) do not depend on the
output values of $f$, and indeed, all the operations commute
with negating the output values of $f$. Therefore,
applying the reduction to $\overline{f}$ is the same as applying
it to $f$ and negating the output values of the result;
but this gives the function $f$, so $f\lesssim' \overline{f}$,
meaning $\overline{f}$ is switchable.
The converse direction follows from replacing the roles of
$f$ and $\overline{f}$. The proof for strongly switchable
functions is identical.
\end{proof}

\begin{corollary}
If $\M$ is strongly well-behaved and $f$ is switchable,
then $\M(\overline{f})=\M(f)$.
If $\M$ is weakly well-behaved and $f$ is strongly switchable,
then we also have $\M(\overline{f})=\M(f)$. 
\end{corollary}

\begin{lemma}\label{lem:Sof}
For every (possibly partial) Boolean function $f$,
the following properties of $\tS\circ f$ hold:
\begin{enumerate}
    \item $\tS\circ f\lesssim f$
    \item $\tS\circ f$ is strongly switchable
    \item $f$ is switchable if and only if
    $f\lesssim' \tS\circ f$
    \item $f$ is strongly switchable if and only if
    $f\lesssim \tS\circ f$.    
\end{enumerate}
\end{lemma}

\begin{proof}
To show that $\tS\circ f\lesssim f$, start with $f$,
add $n(f)$ superfluous bits that can take any value in
$\Dom(f)$ (here $n(f)$ is the input size of $f$),
and impose the promise that $f(x)\ne f(y)$ for every string
$xy$ in the new domain. This gives the function $\tS\circ f$.

To see that $\tS\circ f$ is strongly switchable, observe
that switching the two blocks of $\tS\circ f$ negates
the output value of this function (and this negated function
is actually equal to $\tS\circ \overline{f}$).

For the third item,
if $f$ is switchable, we have $\overline{f}\lesssim' f$.
Apply this reduction to the first block of $\tS\circ f$
(consisting of half the bits). Since the function $\tS\circ f$
is always equal to $f$ applied to the first block,
after the reduction,
the new function is equal to $\overline{f}$ applied to the first
block, and also equal to $\overline{f}$ applied to the second block;
these two blocks are independent other than the condition that
they have the same $f$-value. Now impose the promise that
the two blocks are identical, and remove the resulting duplicate bits.
This gives the function $\overline{f}$. Since $f$ is switchable,
we can convert this to $f$, so $f\lesssim'\tS\circ f$.
The same proof works to show that if $f$ is strongly
switchable then $f\lesssim \tS\circ f$.

Finally, suppose that $f\lesssim'\tS\circ f$.
We will show $f\lesssim' \overline{f}$,
which means $f$ is switchable by \lem{NegationSwitch}.
By the first part of the current lemma, we have
$\tS\circ \overline{f}\lesssim \overline{f}$.
Also, it is easy to see that
$\tS\circ f\lesssim \tS\circ \overline{f}$,
since the only difference between $\tS\circ f$
and $\tS\circ \overline{f}$ is the order of the two
blocks, so rearranging the bits is sufficient to convert between them.
Since we are assuming $f\lesssim'\tS\circ f$
(and since $\lesssim$ is a stronger property than $\lesssim'$),
transitivity gives us $f\lesssim' \overline{f}$.
Similarly, if $f\lesssim \tS\circ f$, the same argument
gives $f\lesssim \overline{f}$. This completes the proof.
\end{proof}

\begin{lemma}\label{lem:SwitchComposition}
For any (possibly partial) Boolean functions
$f$, $g$, $f'$, $g'$, we have:
\begin{enumerate}
    \item if $g'\lesssim g$ then $f\circ g'\lesssim f\circ g$
    \item if $f'\lesssim f$ then $f'\circ g\lesssim f\circ g$
    \item if $g'\lesssim' g$ then $f\circ g'\lesssim' f\circ g$
    \item if $f'\lesssim' f$ and $g$ is switchable,
    then $f'\circ g\lesssim' f\circ g$
    \item if $f'\lesssim' f$ and $g$ is strongly switchable,
    then $f'\circ g\lesssim f\circ g$.
\end{enumerate}
\end{lemma}

\begin{proof}
The first and third items are straightforward:
it is not hard to see that rearranging bits of $g$
rearranges bits of $f\circ g$, restricting $g$ to a promise
restricts $f\circ g$ to a promise, duplicating bits of $g$
duplicates bits of $f\circ g$ (and hence deleting duplicates
of $g$ deletes duplicates of $f\circ g$), adding or removing
superfluous bits of $g$ adds/removes superfluous bits of
$f\circ g$, and negating bits of $g$ negates bits of $f\circ g$.

For the second item,
we need to show that if we can convert $f$ to $f'$ by adding/deleting
duplicated or superfluous bits, rearranging bits, and restricting to a
promise, then we can also convert $f\circ g$ to $f'\circ g$ using
these operations. It is easy to see that rearranging bits of $f$
amounts to rearranging bits of $f\circ g$, and restricting $f$
to a promise amounts to restricting $f\circ g$ to a promise.
Moreover, adding superfluous bits to $f$ adds superfluous bits to $f\circ g$,
and removing superfluous bits from $f$ removes superfluous bits from $f\circ g$
(here a set of bits $S\subseteq[n]$ is \emph{superfluous} if the
value of the function $f(x)$ depends only on $x_{[n]\setminus S}$,
that is, the partial assignments on the bits other than $S$, and moreover,
for any such partial assignment $x_{[n]\setminus S}$, the possible
assignments to $x_S$ that are in the promise are always the same).

It remains to handle the addition and deletion of duplicated bits.
Suppose we add a bit duplication to $f$ to get $f_i$, and consider
$f_i\circ g$. We can construct this from $f\circ g$ by adding
a superfluous set of bits (corresponding to another input to $g$
which is ignored), and then imposing a promise (to ensure
that the added block of bits has the same $g$-value as the
$i$-th input to $g$). Hence $f_i\circ g\lesssim f\circ g$.
Conversely, suppose we delete a duplicated bit to go from $f_i$ to $f$.
This time, we start with $f_i\circ g$ and wish to construct $f\circ g$.
The function $f_i\circ g$ has an extra input to $g$; we will
impose the promise that this extra input is identical to the $i$-th
input to $g$, which makes the extra input bits duplicated bits,
and then we will delete the duplicated bits from the resulting
function. This completes the proof of the second item.

For the last two items, we need to show how to handle negating
bits of $f$. Note that negating bits of $f$ corresponds
to switching the $g$-values of blocks of $f\circ g$.
If $g$ is (strongly) switchable, we can convert it to $\overline{g}$
via the above operations (excluding bit negations in the case
of strongly switchable $g$). Applying this to a single
block flips the corresponding bit of the outer function $f$
inside the composition $f\circ g$. The desired result follows.
\end{proof}

\begin{corollary}
Let $f$ and $g$ be (possibly partial) Boolean functions.
Then
\begin{enumerate}
    \item If $f$ is strongly switchable, so is $f\circ g$.
    \item If $f$ and $g$ are both switchable, so is $f\circ g$.
    \item If $f$ is switchable and $g$ is strongly switchable,
    then $f\circ g$ is strongly switchable.
\end{enumerate}
\end{corollary}

\begin{proof}
If $f$ is strongly switchable, then $\overline{f}\lesssim f$,
so by \lem{SwitchComposition} we have
$\overline{f}\circ g\lesssim f\circ g$. Since
$\overline{f}\circ g$ is the same function as $\overline{f\circ g}$,
the first item follows. The next two items follow from
\lem{SwitchComposition} in a similar way.
\end{proof}

We note that $\M(f\circ g)$ can be viewed as a measure
of $f$ (with fixed $g$) or as a measure of $g$ (with
fixed $f$). The following property follows.

\begin{corollary}
Suppose $\M(f)$ is weakly well-behaved on composition-closed
$\cA$. Then $\M(f\circ g)$ is a weakly well-behaved measure
of $f$ and a weakly well-behaved measure of $g$.

Moreover, if $\M(f)$ is strongly well-behaved,
then $\M(f\circ g)$ is strongly well-behaved as a function
of $g$, and if additionally $g$ is switchable,
then $\M(f\circ g)$ is strongly well-behaved as a function
of $f$.
\end{corollary}

\begin{proof}
This follows immediately from \lem{SwitchComposition}
and \lem{WellBehavedReductions}.
\end{proof}

\subsection{Block sensitivity bounds}

\begin{definition}
For $n\in\bN$, let the ``promise-OR'' function
$\PrOR_n$ be the function on $n$ bits
with $\Dom(\PrOR_n)=\{x:|x|=0\text{ or }|x|=1\}$ defined by
$\PrOR_n(0^n)=0$ and $\PrOR_n(x)=1$ if $|x|=1$.
\end{definition}

\begin{theorem}\label{thm:bsReduction}
For any (possibly partial) Boolean function $f$,
$\PrOR_{\bs(f)}\circ \tS\lesssim f$.
Therefore, if $\M$ is any weakly well-behaved measure,
then $\M(f)\ge \M(\PrOR_{\bs(f)}\circ \tS)$.
\end{theorem}

\begin{proof}
Let $k=\bs(f)$, and let $n=n(f)$ be the input size of $f$.
Recall that $k=\max_{x\in\Dom(f)}\bs_x(f)$.
Let $x\in\Dom(f)$ satisfy $\bs_x(f)=k$, and let
$B_1,B_2,\dots,B_k$ be disjoint sensitive blocks of $x$.
Restrict $f$ to the promise set $\{x\}\cup\{x^{B_j}:j\in[k]\}$,
that is, the string $x$ together with the block flips of $x$
($k+1$ strings in total). Any bit $i\in[n]$ which is not in
any block will be constant after restricting to this promise,
and hence superfluous; remove all such bits. Moreover,
within each block $B_j$, all bits which take value $0$
in the string $x$ are duplicates, and all bits which take
value $1$ are duplicates; after removing duplicates,
the block $B_j$ becomes either size $1$ or size $2$.
If the block $B_j$ becomes a single bit, we can add a superfluous
bit to the input and add the promise that this new bit
will always equal the negation of the single bit in $B_j$;
this reduces to the case where $B_j$ has two bits that are
always either $01$ or $10$.

Finally, we can rearrange the bits so that all the blocks
are contiguous and, in the input that was originally $x$,
all blocks take the form $01$. This means that the input
$x$ becomes $(01)^k=010101\dots01$, while each input
$x^{B_j}$ becomes the same string with the $j$-th pair $01$
becoming $10$ instead. This is precisely the function
$\PrOR_k\circ \tS$, so we have shown $\PrOR_k\circ S\lesssim f$,
as desired.
\end{proof}

\begin{lemma}\label{lem:IS}
For all $n\in\bN$, $\tI\lesssim \PrOR_n$.
We also have $\tI\lesssim' \tS$.
\end{lemma}

\begin{proof}
The first claim follows by restricting to the promise
$\{0^n,10^{n-1}\}$ and removing the $n-1$ constant (superfluous)
bits. For the second claim, negate the second bit of $\tS$;
then the two bits are duplicate, and removing duplicates gives
$\tI$.
\end{proof}

\begin{corollary}\label{cor:bsComposition}
Let $f$ and $g$ be (possibly partial) functions which
are not constant, and let $\M$ be weakly well-behaved. Then
\begin{enumerate}
\item $\M(f\circ g)\ge \M(f\circ\PrOR_{\bs(g)}\circ\tS)
\ge\M(f\circ \tS)$
\item $\M(f\circ g)\ge \M(\PrOR_{\bs(f)}\circ\tS\circ g)
\ge\M(\tS\circ g)$
\item If $\M$ is strongly well-behaved,
$\M(f\circ g)\ge \M(f\circ\PrOR_{\bs(g)})\ge\M(f)$
\item If $\M$ is strongly well-behaved and $g$ is switchable,
$\M(f\circ g)\ge \M(\PrOR_{\bs(f)}\circ g)\ge\M(g)$.
\end{enumerate}
\end{corollary}

\begin{proof}
All of these follow easily from \thm{bsReduction},
\lem{SwitchComposition}, and \lem{IS}.
\end{proof}

\paragraph{Remark.} \cor{bsComposition} shows that the
switch function $\tS$ is the easiest non-constant function
for well-behaved measures in a fairly strong sense
(specifically, $\M(f\circ g)\ge \M(f\circ \tS)$
and $\M(f\circ g)\ge \M(\tS\circ g)$ for all non-constant
$f$ and $g$).

\section{Properties of recursive composition}

\subsection{Reasonably bounded measures}

We introduce some definitions regarding the composition
behavior of measures.

\begin{definition}[Reasonably upper bounded (RUBO, RUBI)]
Let $\cA\subseteq \cF$ be a composition-closed class of
functions, and let $\M$ be a measure on $\cA$. We say
$\M$ is \emph{reasonably upper bounded on the outside
(RUBO)} if there is some measure $\N\colon\cA\to[0,\infty)$
such that for all $f,g\in\cA$, $\M(f\circ g)\le \N(f)\M(g)$.

Similarly, we say that $\M$ is
\emph{reasonably upper bounded on the inside (RUBI)}
if there is some measure $\N$ such that
$\M(f\circ g)\le \M(f)\N(g)$ for all $f,g\in\cA$.

If the measure $\N$ satisfies $\N(f)\le \M(f)^{1+o(1)}$,
we say the corresponding bound holds \emph{nearly linearly}.
\end{definition}

We wish to also define a reasonably lower bounded
measure as one satisfying $\M(f\circ g)\ge \N(f)\M(g)$.
However, this definition is trivial if $\N(f)=0$,
yet we cannot require $\N(f)>0$ always since
$\M(f\circ g)$ may equal $0$. To make it nontrivial,
we require $\N(f)$ to be nonzero unless $\M(f\circ g)$ is zero.

\begin{definition}[Reasonably lower bounded (RLBO, RLBI)]\label{def:RLB}
Let $\cA\subseteq\cF$ be a composition-closed class
of functions, and let $\M$ be a measure on $\cA$.
We say $\M$ is \emph{reasonably lower bounded on the outside
(RLBO)} if there is some measure $\N$ on $\cA$ such that
for all $f,g\in\cA$, $\M(f\circ g)\ge \N(f)\M(g)$, and
additionally, $\N(f)=0$ implies $\M(f\circ g)=0$ for all
$f,g\in\cA$.

Similarly, we say $\M$ is
\emph{reasonably lower bounded on the inside (RLBI)}
if there is some measure $\N$ such that $\M(f\circ g)\ge \M(f)\N(g)$
and $\N(g)=0\Rightarrow \M(f\circ g)=0$ for all $f,g\in\cA$.

If the measure $\N$ satisfies $\N(f)\ge \M(f)^{1-o(1)}$,
we say the corresponding bound holds nearly linearly.
\end{definition}

We note that RLBI holds for all well-behaved measures.

\begin{lemma}\label{lem:RLBI}
Suppose $\M$ is strongly well-behaved and $\M(f)=0$ for all
constant functions $f$. Then $\M$ satisfies RLBI.
\end{lemma}

\begin{proof}
Recall \cor{bsComposition}, which gives $\M(f\circ g)\ge\M(f)$
whenever $g$ is not constant. We can then define
$\N(g)$ to be $0$ for all constant $g$ and $\N(g)=1$
for all non-constant $g$; then $\M(f\circ g)\ge \M(f)\N(g)$.
Finally, if $\N(g)=0$, then $g$ is constant, so $f\circ g$
is constant, which means $\M(f\circ g)=0$.
\end{proof}

We also have the following.

\begin{theorem}
The measures $\D(f)$, $\R(f)$, $\R_0(f)$, $\Q(f)$,
$\s(f)$, $\fbs(f)$, $\C(f)$, $\deg(f)$,
and $\adeg(f)$ all satisfy RLBI, and all satisfy
RUBO nearly linearly.
\end{theorem}

\begin{proof}
It is straightforward to check that all of these measures
except for $\s(f)$ are strongly well-behaved, and all are
equal to $0$ for constant functions; hence by \lem{RLBI},
they all satisfy RLBI. It is also not hard to check
that $\s(f\circ g)\ge \s(f)$ if $\s(g)>0$, so $\s$
also satisfies RLBI.

The algorithmic measures $\D$, $\R$, $\R_0$, and $\Q$
can all compute $f\circ g$ by running the algorithm for $f$,
and whenever that algorithm queries a bit of $f$, run the corresponding
algorithm for $g$. The measures with error, $\R$ and $\Q$,
will need to apply error reduction to the subroutine for $g$
to reduce the error down to $o(1/\R(f))$ and $o(1/\Q(f))$
respectively; this requires an overhead of $O(\log \R(f))$
and $O(\log\Q(f))$ respectively. From this, it is easy to see
that $\D(f\circ g)\le \D(f)\D(g)$, $\R_0(f\circ g)\le\R_0(f)\R_0(g)$,
$\R(f\circ g)\le O(\R(f)\log\R(f))\cdot\R(g)$,
and $\Q(f\circ g)\le O(\Q(f)\log\Q(f))\cdot \Q(g)$;
hence these measures satisfy RUBO nearly linearly.
A similar strategy works for the degree measures $\deg$ and $\adeg$,
since polynomials can be composed and error reduction for polynomials
works similarly to error reduction for algorithms.

This also works for certificates, since to certify
$f\circ g$ it suffices to certify the outer function $f$
and then for each bit of the certificate, certify the inner
function $g$; thus $\C(f\circ g)\le \C(f)\C(g)$.
For fractional block sensitivity $\fbs$, we use its dual
form as fractional certificate complexity \cite{Aar08,KT16}.
We can specify a fractional certificate for $f\circ g$
by picking a fractional certificate for $f$ and composing
it with a fractional certificate for $g$, so
that $\fbs(f\circ g)\le \fbs(f)\fbs(g)$.
Finally, one can directly show that $\s(f\circ g)\le \s(f)\s(g)$,
since in any input to $f\circ g$ the only sensitive
bits correspond to sensitive bits for an input to $g$
which in turn lies in a sensitive bit for the input to $f$.
This completes the proof.
\end{proof}

\paragraph{Remark.} Notably absent from the above list is
block sensitivity $\bs(f)$, which does not satisfy RUBO
nearly linearly. The composition behavior of block sensitivity
was investigated in \cite{Tal13,GSS16}.

The above motivates the following definition.

\begin{definition}[Composition bounded]
A measure $\M$ on composition-closed $\cA$ is said to
be \emph{composition bounded} if it satisfies RLBI
and satisfies RUBO nearly linearly.
\end{definition}

By this definition, all the above measures (except $\bs(f)$)
are composition bounded.

\subsection{Convergence theorem}

\begin{theorem}\label{thm:FormalLimit}
Let $\mathcal{A}\subseteq\mathcal{F}$ be closed under
composition, and let $M$ be a
composition-bounded measure on $\mathcal{A}$.
Then $M^*(f)$ converges for all $f\in\mathcal{A}$.
\end{theorem}

To prove this theorem, we will need the following lemma.

\begin{lemma}\label{lem:subsequence}
Suppose $M$ is reasonably upper bounded
(either on the outside or on the inside)
and also reasonably lower bounded (either on the
outside or on the inside).
Then there is a sequence
of positive integers $k_1,k_2,\dots$ such that the following conditions hold:
\begin{enumerate}
\item the sequence is increasing: for all $\ell\in\bN^+$, we have $k_\ell<k_{\ell+1}$
\item the integers divide each other: for all $\ell\in\bN^+$, we have $k_\ell|k_{\ell+1}$
\item the even elements of the sequence are such that
$M(f^{k_{2\ell}})^{1/k_{2\ell}}\to M^{\underline{*}}(f)$ as $\ell\to\infty$
\item the odd elements of the sequence are such that
$M(f^{k_{2\ell+1}})^{1/k_{2\ell+1}}\to M^{\overline{*}}(f)$ as $\ell\to\infty$.
\end{enumerate}
\end{lemma}

\begin{proof}
We describe how to build the infinite sequence $k_1,k_2,\dots$
one at a time. We start with a sequence of length $0$.
Now suppose, at an intermediate step, we already have
$k_1,k_2,\dots,k_{n-1}$ in the sequence (these are increasing integers
that divide each other). For any fixed $\epsilon>0$, we show how to add
$k_n$ such that $k_n>k_{n-1}$, $k_{n-1}|k_n$, and such that
$\M(f^{k_n})^{1/k_n}\ge \M^{\overline{*}}(f)(1-\epsilon)$. We will similarly
show how to add $k_n$ such that
$\M(f^{k_n})^{1/k_n}\le \M^{\underline{*}}(f)(1+\epsilon)$ instead.
We can then alternate adding an approximation to $\M^{\overline{*}}(f)$
or $\M^{\underline{*}}(f)$ to the sequence, and decrease $\epsilon$
as we progress.

To add $k_n$ with $\M(f^{k_n})^{1/k_n}\ge \M^{\overline{*}}(f)(1-\epsilon)$,
we start by noting that there are infinitely many integers $k$ such that
$|\M(f^k)^{1/k}- \M^{\overline{*}}(f)|<\epsilon/2$. For any such $k\ge 2k_{n-1}$,
write $k=ck_{n-1}+r$ with $r\in[0,k_{n-1}-1]$. Then
$\M(f^{ck_{n-1}+r})$ is within a factor of $\N(f)^r$ of $\M(f^{ck_{n-1}})$,
so $\M(f^{ck_{n-1}})^{1/ck_{n-1}}$ is within a factor of $\N(f)^{r/ck_{n-1}}$
of $\M(f^k)^{1/(k-r)}$, which is itself within a factor of
$\M(f^k)^{1/(k-r)-1/k}=(\M(f^k)^{1/k})^{r/(k-r)}$ of $\M(f^k)^{1/k}$.
The expression $(\M(f^k)^{1/k})^{r/(k-r)}$ is closer to $1$ than
$(\M(f^k)^{1/k})^{1/c}$, and $\N(f)^{r/ck_{n-1}}$ is closer to $1$ than $\N(f)^{1/c}$.
These are closer to $1$ than $(\M(f^k)^{1/k})^{1/(k-k_{n-1})}$ and $\N(f)^{1/(k-k_{n-1})}$
respectively. Note that $\M(f^k)^{1/k}\ge \M^{\overline{*}}(f)-\epsilon/2$;
if $\M^{\overline{*}}(f)>0$ and $\epsilon$ is small enough, this means
that $\M(f^k)^{1/k}$ is bounded between two positive constants, so by picking
$k$ large enough, $(\M(f^k)^{1/k})^{1/(k-k_{n-1})}$ can be made arbitrarily
close to $1$, as can $\N(f)^{1/(k-k_{n-1})}$. Hence by choosing $k$ sufficiently
large, we can set $k_n=ck_{n-1}$ and conclude that $\M(f^{k_n})^{1/k_n}$ is
within $\epsilon$ of $\M^{\overline{*}}(f)$.
On the other hand, if $\M^{\overline{*}}(f)=0$, then $\M(f^k)^{1/k}$ can be made
arbitrarily close to $0$. In that case,
$\M(f^k)^{1/k}\cdot (\M(f^k)^{1/k})^{\pm 1/(k-k_{n-1})}
= (\M(f^k)^{1/k})^{1\pm 1/(k-k_{n-1})}$ can also be made arbitrarily close to $0$,
and $\N(f)^{1/(k-k_{n-1})}$ can be made arbitrarily close to $1$.
Hence we can still pick $k_n=ck_{n-1}$ to approximate $\M^{\overline{*}}(f)$
to error $\epsilon$.

These arguments work for $\M^{\underline{*}}(f)$ as well, so we can always find
a value $k$ which is a multiple of anything we want and which approximates
either  $\M^{\underline{*}}(f)$ or  $\M^{\overline{*}}(f)$ to any desired
error $\epsilon$ of our choice. By alternating between adding
an approximation to $\M^{\underline{*}}(f)$ and to $\M^{\overline{*}}(f)$
to our sequence, and by decreasing the chosen value of $\epsilon$ to $0$,
we can construct the desired sequence.
\end{proof}

\begin{proof}[Proof of \thm{FormalLimit}]
We start with a few edge cases. First, we handle
the case where $\M(f^k)=0$ for infinitely many values of $k$.
In this case, we claim that $\M(f^k)=0$ for all $k\ge 1$.
To see this, suppose not, and let $k_1<k_2$ be such that
$\M(f^{k_1})>0$ and $\M(f^{k_2})=0$. Now, since $M$ is RLBI,
there is a measure $\N$ such that
$\M(f\circ g)\ge \M(f)\N(g)$ for all $f$ and $g$;
using $f^{k_2-1}$ and $f$ as the two functions,
we get $M(f^{k_2})\ge \M(f^{k_2-1})\N(f)$, and by repeating we get
$\M(f^{k_2})\ge \N(f)\M(f^{k_2-1})\ge \N(f)^2\M(f^{k_2-2})
\ge\dots \ge \N(f)^{k_2-k_1}\M(f^{k_1})$.
Since $\M(f^k_2)=0$ and $\M(f^{k_1})>0$, it follows that
$\N(f)=0$. However, the condition of \defn{RLB}
says that $\N(f)=0$ implies $\M(f\circ g)=0$ for all $g$
(including $g=\tI$).
In particular, we must have $\M(f^k)=0$ for all $k\ge 1$,
which implies that $\M^*(f)$ converges to $0$.

Next, consider the sequence $k_1,k_2,\dots$ from \lem{subsequence}.
Note that by the RUBO/RLBI properties of $\M$,
it follows that $\M^{\overline{*}}(f)\le \N(f)$
and $\M^{\underline{*}}(f)\ge \N'(f)$ for possibly
different measure $\N$ and $\N'$, so in particular, these
values are both finite. Hence to show that $\M^*(f)$ converges,
we just need to show that $\M^{\overline{*}}(f)=\M^{\underline{*}}(f)$,
which is the same as showing that the even and odd terms
of the sequence $[\M(f^{k_i})^{1/k_i}]_{i=1}^\infty$
converge to the same value.

Since we may assume $\M(f^k)$ is only zero a finite number
of times, we can pick $i$ large enough so that $\M(f^{k_j})>0$
for all $j\ge i$. Then $k_j$ is a multiple of $k_i$
for all $j\ge i$. Let $g=f^{k_i}$; then the sequence
$[\M(f^{k_j})^{1/k_j}]_{j=i}^\infty$ can be written
$[\M(g^{m_\ell})^{1/k_im_\ell}]_{\ell=0}^\infty$
where $m_\ell=k_{\ell+i}/k_i$. From this it is not hard
to see that $\M^*(f)$ converges if and only if $\M^*(g)$
converges; by replacing $f$ with $g$ if necessary,
we may assume that $\M(f^k)>0$ for all $k\ge 1$,
and that we still have a valid sequence
$k_1,k_2,\dots$ as in \lem{subsequence},
yet this time all values $\M(f^{k_i})$ are nonzero.

It suffices to show that $\M^{\overline{*}}(f)\le \M^{\underline{*}}(f)+\epsilon$ for any fixed $\epsilon>0$.
By the nearly-linear RUBO property, we have
$\M(f^k)\le \M(f)^{1+o(1)}\M(f^{k-1})$ for all
$k\ge 2$. Let $\delta>0$ be small enough that
$(\M^{\underline{*}}(f)+\epsilon/3)^{1+\delta}
\le \M^{\underline{*}}(f)+2\epsilon/3$.
From \defn{littleo}, let $C\ge 1$ be such that
$\M(f^k)\le \M(f^\ell)^{1+\delta}\M(f^{k-\ell})$
whenever $k\ge 2$, $1\le \ell\le k-1$, and
$\M(f^k)> C\M(f^\ell)\M(f^{k-\ell})$.
Then for each $k\ge 2$, we have
$\M(f^k)\le \M(f^\ell)\M(f^{k-\ell})\max\{C,\M(f^\ell)^\delta\}$.
Similarly, $\M(f^{k-\ell})\le \M(f^\ell)\M(f^{k-2\ell})\max\{C,\M(f^\ell)^\delta\}$.
Putting it together, we get that whenever $\ell$
divides $k$, we have
\begin{align*}
\M(f^k)&\le \M(f^\ell)\M(f^{k-\ell})\max\{C,\M(f^\ell)^\delta\}
\\&\le \M(f^\ell)^2\M(f^{k-2\ell})\max\{C,\M(f^\ell)^\delta\}^2
\\&\le ...
\\&\le \M(f^\ell)^{k/\ell}\max\{C,\M(f^\ell)^{\delta}\}^{k/\ell}.
\end{align*}
Hence whenever $\ell$ divides $k$, we can write
\[\M(f^k)^{1/k}\le \M(f^\ell)^{1/\ell}\max\{C^{1/\ell},(\M(f^\ell)^{1/\ell})^{\delta}\}.\]
Now, pick an even term $k_i$ large enough so that
$\M(f^{k_i})^{1/k_i}\le \M^{\underline{*}}(f)+\epsilon/3$
and so that $C^{1/k_i}(\M^{\underline{*}}(f)+\epsilon/3)\le
\M^{\underline{*}}(f)+2\epsilon/3$.
Also, pick an odd term $k_j$ larger than $k_i$, and large enough
so that $\M(f^{k_j})^{1/k_j}\ge \M^{\overline{*}}(f)-\epsilon/3$.
Then set $k=k_j$ and $\ell=k_i$ in the above, and write
\begin{align*}
\M^{\overline{*}}(f)&\le \M(f^k)^{1/k}+\epsilon/3
\\&\le \epsilon/3+\M(f^\ell)^{1/\ell}\max\{C^{1/\ell},(\M(f^\ell)^{1/\ell})^{\delta}\}
\\&\le \epsilon/3+\max\{C^{1/\ell}(\M^{\underline{*}}(f)+\epsilon/3),
                        (\M^{\underline{*}}(f)+\epsilon/3)^{1+\delta}\}
\\&\le \epsilon+\M^{\underline{*}}(f).
\end{align*}
Since $\epsilon$ was arbitrary, we must have
$\M^{\overline{*}}(f)=\M^{\underline{*}}(f)$, as desired.
\end{proof}

\subsection{Properties of composition limits}

We list some basic properties of composition limits.

\begin{lemma}
If $\M$ is a weakly well-behaved measure and $\M^*$ converges,
then $\M^*$ is also weakly well-behaved.
\end{lemma}

\begin{proof}
We need to show that if $f'\lesssim f$, then $\M^*(f')\le\M^*(f)$
(see \lem{WellBehavedReductions}).
It suffices to show that $\M((f')^k)\le\M(f^k)$ for all $k\in\bN$.
Using \lem{SwitchComposition}, we can replace each
copy of $f$ in $f^k=f\circ f\circ\dots\circ f$ by $f'$ one at a time,
until we get $(f')^k=f'\circ\dots\circ f'$. Each time we replace
$f$ by $f'$, we get a function that is smaller than or equal to
the previous one in the $\lesssim$ order. Hence
$(f')^k\lesssim f^k$, and the desired result follows.
\end{proof}

\begin{lemma}\label{lem:kstar}
If $\M^*(f)$ converges, then $\M^*(f^k)$ converges for each $k\in\bN$ and
$\M^*(f^k)=\M^*(f)^k$.
\end{lemma}

\begin{proof}
We have 
\[\M^*(f^k)=\lim_{n\to\infty}\M((f^k)^n)^{1/n}
=\lim_{n\to\infty}(\M(f^{kn})^{1/kn})^k
=(\lim_{n\to\infty}\M(f^{kn})^{1/kn})^k
=\M^*(f)^k,\]
where we exchanged a limit with the continuous function $x^k$
and we used the fact that a subsequence converges to the same
limit as a convergent sequence.
\end{proof}

\begin{corollary}\label{cor:StarStar}
If $\M^*(f)$ converges, $\M^{**}(f)$ also converges and
$\M^{**}(f)=\M^*(f)$.
\end{corollary}

\begin{proof}
We have
$\M^{**}(f)=\lim_{k\to\infty}\M^*(f^k)^{1/k}$,
but $\M^*(f^k)=\M^*(f)^k$ by \lem{kstar}, so this is the constant
sequence $\M^*(f)$ which converges to $\M^*(f)$.
\end{proof}

\begin{lemma}
Suppose $\M$ is composition bounded on
a composition-closed class $\mathcal{A}\subseteq\mathcal{F}$.
Then $\M^*(f\circ g)=\M^*(g\circ f)$ for all $f,g\in\mathcal{A}$.
\end{lemma}

\begin{proof}
Suppose that $\N(f)\ne 0$, where $\N$ is the measure
from the RLBI property. Write
\[\M((f\circ g)^k)=\M(f\circ (g\circ f)^{k-1} \circ g)
\le \frac{\M(f)^{1+o(1)}}{\N(f)}\M((g\circ f)^k),\]
where we used the nearly-linear RUBO property to bound
$\M(f\circ (g\circ f)^{k-1}\circ g)\le \M(f)^{1+o(1)}
\M((g\circ f)^{k-1}\circ g)$
and we used the RLBI property to bound
$\M((g\circ f)^{k-1}\circ g\circ f)\ge
\M((g\circ f)^{k-1}\circ g)\N(f)$.
Raising both sides to the power $1/k$ and taking limits gives
the desired result.

Next, suppose that $\N(g)\ne 0$. We can use the same proof
starting with $\M((g\circ f)^k)$ instead of starting with
$\M((f\circ g)^k)$.

Finally, suppose $\N(f)=\N(g)=0$.
Then $\M(h\circ f)=0$ for all $h\in\cA$,
and in particular $\M((g\circ f)^k)=0$ for all $k$.
Similarly, $\M((f\circ g)^k)=0$ for all $k$.
This means $\M^*(g\circ f)=0=\M^*(f\circ g)$, as desired.
\end{proof}

\begin{lemma}\label{lem:StarCompare}
If $\M(f)\le \N(f)^{c+o(1)}$ for some constant $c$
and $\M^*(f)$ and $\N^*(f)$ both converge with $\M^*(f)\ge 1$,
then $\M^*(f)\le \N^*(f)^c$.
\end{lemma}

\begin{proof}
By the definition of the little-o notation \defn{littleo},
for any fixed $\delta>0$, there exists a constant $C>0$
for which $\M(f)\le \max\{\N(f)^{c+\delta},C\}$ holds for all $f$.
Then $\M(f^k)^{1/k}\le \max\{\N(f^k)^{(c+\delta)/k},C^{1/k}\}$,
and taking limits, $\M^*(f)\le \max\{\N^*(f)^{c+\delta},1\}$.
Since this holds for all $\delta>0$, we get
$\M^*(f)\le \max\{\N^*(f)^c,1\}$, as desired.
\end{proof}

\begin{lemma}\label{lem:StarBounds}
If $\M(f\circ g)\le \N(f)\M(g)$ for all $f$ and $g$
and $\M^*(f)$ converges,
then $\M^{*}(f)\le \N(f)$. Similarly,
if $\M(f\circ g)\le \M(f)\N(g)$ for all $f$ and $g$,
then $\M^{*}(f)\le \N(f)$.

This also works in the lower bound direction:
if either $\M(f\circ g)\ge \N(f)\M(g)$ for all $f$ and $g$
or else $\M(f\circ g)\ge \M(f)\N(g)$ for all $f$ and $g$,
then $\M^{*}(f)\ge \N(f)$ so long as
$\M(f^k)$ is not always $0$ (and so long as $\M^*(f)$ converges).
\end{lemma}

\begin{proof}
If $\M(f\circ g)\le \N(f)\M(g)$, then
$\M(f^k)\le \N(f)\M(f^{k-1})\le \N(f)^2\M(f^{k-2})
\le\dots\le \N(f)^k\M(\tI)$.
Then $\M(f^k)^{1/k}\le \N(f)\M(\tI)^{1/k}$.
As $k\to\infty$, $\M(\tI)^{1/k}$ converges either to $0$
or to $1$, and in both cases we get $\M^*(f)\le \N(f)$,
as desired. This also works when $\M(f\circ g)\le \M(f)\N(g)$.

For the lower bound, the same argument gives
$\M(f^k)\ge \N(f)^{k-\ell}\M(f^{\ell})$ for any $k\ge \ell$.
Pick $\ell$ such that $\M(f^\ell)\ne 0$. Then
$\M(f^k)^{1/k}\ge \N(f)^{1-\ell/k}\M(f^\ell)^{1/k}$,
and as $k\to\infty$, we get $\M^*(f)\ge \N(f)$, as desired.
\end{proof}

\begin{corollary}
$\R^*(f)=O(\R(f)\log\R(f))$. Moreover,
$\R^*(f)=\Omega(\noisyR(f)+\LR(f))$, where the measures
$\noisyR(f)$ and $\LR(f)$ are defined in \cite{BB20b}
and \cite{BBGM22} respectively.
\end{corollary}

\begin{proof}
This follows from \lem{StarBounds} when combined with the observations
 $\R(f\circ g)=O(\R(f)\log\R(f)\cdot \R(g))$,
$\R(f\circ g)=\Omega(\noisyR(f)\R(g))$ \cite{BB20b}, and
$\R(f\circ g)=\Omega(\R(f)\LR(g))$ \cite{BBGM22}, which hold
for all partial functions $f$ and $g$.
\end{proof}

\begin{lemma}\label{lem:StarArithmetic}
Let $\M$ and $\N$ be measures with convergent composition limits,
and let $c>0$ be a constant. Then
\begin{enumerate}
    \item $c\M$ has a convergent composition limit and $(c\M)^*=\M^*$
    \item $\M^c$ has a convergent composition limit and 
    $(\M^c)^*=(\M^*)^c$
    \item $\M+\N$ has a convergent composition limit and 
    $(\M+\N)^*=\max\{\M^*,\N^*\}$
    \item $\M\cdot \N$ has a convergent composition limit and 
    $(\M\cdot \N)^*=\M^*\cdot\N^*$.
\end{enumerate}
\end{lemma}

\begin{proof}
All of these follow immediately from properties of limits.
The trickiest one is $(\M+\N)^*=\max\{\M^*,\N^*\}$,
which is equivalent to showing
$\lim_{k\to\infty} (\M(f^k)+\N(f^k))^{1/k}=\max\{\M^*(f),\N^*(f)\}$.
This is not hard to show by using
$\M(f^k)+\N(f^k)\le 2\max\{\M(f^k),\N(f^k)\}$.
\end{proof}

\section{Composition limits for Las Vegas algorithms}
\label{sec:LasVegas}

Define the measure $\Q_{\C}(f)$ to be the number of queries
required by a quantum algorithm that finds a certificate.
That is, the quantum algorithm must output a certificate
$c$ certifying the value $f(x)$ of the input $x$, and
must succeed in producing $c$ with probability at least
$1/2$ (when it fails, it should output a failure symbol
$\bot$). One subtlety of the definition is that we require
the certificate size of the output to contribute to the cost;
that is, the cost of returning the certificate $c$ should
be the number of quantum queries used, plus $|c|$, the
number of bits revealed by $c$. The intuition is that
a classical algorithm will verify the returned certificate
using an additional $|c|$ queries.
Note that this ensures $\Q_{\C}(f)\ge \C(f)$.

We note that since such an algorithm can
be repeated upon failure, it can be used to produce
a Las Vegas style quantum algorithm for $f$ which
at most $2\Q_{\C}(f)$ expected queries.
We further note that the randomized version of this measure
is the same as $\R_0(f)$ up to constant factors,
since it is well-known that a
randomized query algorithm which makes zero error must
find a certificate in the input.
(To see why this is true, note that if
a certain run of an $\R_0(f)$ type algorithm terminated
without finding a certificate and claimed to know the
value of $f(x)$ on the input $x$, then by definition
of certificates, there is some $y$-input in the domain
of $f$ which is consistent with the queried bits
and for which $f(y)\ne f(x)$; then if the same algorithm
were to be run on $y$ instead of $x$, there would be a
non-zero probability that it would reach the same leaf
and provide the same output, contradicting the zero-error
property.)

These certificate-finding quantum algorithms are therefore
one possible way to define a zero-error quantum algorithm.

\begin{lemma}
$\Q_{\C}(f)$ is strongly well-behaved, composition bounded,
and $\Q_{\C}^*(f)$ converges for all (possibly partial) $f$.
\end{lemma}

\begin{proof}
It is easy to see that $\Q_{\C}(f)$ is strongly well-behaved:
if the certificate returned is minimal, then both the
certificate size and the number of queries used in an
algorithm are invariant under renaming indices, adding
superfluous bits, duplicating bits, and alphabet renaming
(flipping bits from $0$ to $1$); additionally, the measure
does not increase under restrictions to a promise.

Note also that $\Q_{\C}(f)=0$ if $f$ is constant. This is
because we can always return the empty certificate (making
no queries). By \lem{RLBI}, it follows that $\Q_{\C}(f)$
satisfies RLBI. It remains to show that it satifies RUBO
nearly linearly, which makes it composition bounded
and hence by \thm{FormalLimit} $\Q_{\C}^*(f)$ converges.
In other words, we must show $\Q_{\C}(f\circ g)\le \Q_{\C}(f)^{1+o(1)}\Q_{\C}(g)$.

This follows from the usual
composition of quantum algorithms, as follows: start
with a optimal algorithm for $\Q_{\C}(g)$, and amplify
it by repeating it $O(\log \Q_{\C}(f))$ times to reduce
its failure probability to, say, $(10\Q_{\C}(f))^{-10}$.
Then use this amplified algorithm for $g$ in place of
the quantum oracle calls for the bits of the input to $f$,
and run the algorithm for $\Q_{\C}(f)$ (repeated
twice for amplification purposes). When we run the algorithm
for $\Q_{\C}(f)$ in this way, we make sure to make
classical queries (to those quantum $g$ subroutines)
for the certificate of $f$ that we end up finding.
Since each $g$ subroutine returns a certificate for
that copy of $g$, we end up with a certificate for $f$
and, for each of its bits, a certificate for the 
corresponding copy of $g$; together, these form
a certificate for $f\circ g$, which we then return.

Since the failure
probability of the subroutines for $g$ is so small,
the work states in the resulting algorithm for $f$
(when run on the subroutines
for $g$) are close in trace distance to the work states
of the algorithm for $f$ when run on clean oracle queries.
It is then easy to see that the total failure probability
of the resulting algorithm is at most $1/4+o(1)\le 1/2$
(the $1/4$ came from the fact that we ran the $\Q_{\C}(f)$
twice, each time with a $1/2$ chance of failure; the $o(1)$
is the contribution of the failure of the $g$ subroutines).
The total number of queries used is then
$O(\Q_{\C}(f)\Q_{\C}(g)\log\Q_{\C}(f))$.
This gives the desired upper bound on $\Q_{\C}(f\circ g)$,
completing the proof.
\end{proof}

We now prove our main result.

\begin{theorem}
\[\R_0^*(f)=\max\{\R^*(f),\C^*(f)\},\]
\[\Q_{\C}^*(f)=\max\{\Q^*(f),\C^*(f)\}.\]
\end{theorem}

\begin{proof}
We show this for $\Q_{\C}(f)$. The argument for randomized
query complexity will be similar.
Fix a (possibly partial) Boolean function $f$ defined
on $n$ bits. We may assume $n\ge 2$ since the $n=1$ case
is trivial. We also assume $f$ is not constant.

We recursively define a quantum algorithm finding
a certificate for $f^k$, which we denote $A_k$ for
each $k$; each such algorithm $A_k$ will find a certificate
for $f^k$ with probability at least $1/2$.
The worst case number of queries used by $A_k$ will
be denoted $q_k$.
The base case is $k=1$; for the function $f^1=f$
we simply bound $q_1\le n$.

Suppose now that $k>0$ and $A_{k-1}$ has been defined.
The algorithm $A_k$ works as follows. Consider the top-most
copy of $f$ in the composition $f^k$; this copy has $n$
bits, each of which is $f^{k-1}$ applied to disjoint inputs.
For each of those $n$ copies of $f^{k-1}$, we start by
estimating its output value using a Monte Carlo algorithm.
Specifically, we use the best possible bounded-error
quantum algorithm for $f^{k-1}$, which uses
$\Q(f^{k-1})$ queries, and we amplify it (by
repeating several times and taking majority vote)
until its worst-case error drops from $1/3$ to $1/4n$.
This uses $O(\Q(f^{k-1})\log n)$ quantum queries per copy
of $f^{k-1}$, and there are $n$ copies. Note that since
$n\ge 2$, we have $\log n>0$, and the constant in the
big-O notation can be assumed to be multiplicative.
In other words, there is a universal constant $C$
(independent of $f$, $k$, and $n$) such
that using at most $C\Q(f^{k-1})n\log n$ queries,
we can produce estimates of all $n$ inputs to the topmost
copy of $f$. Moreover, each estimated bit is wrong
with probability at most $1/4n$, and by the union bound,
they're all correct except with probability at most $1/4$.

Let $x\in\B^{n^k}$ denote the input to $f^k$, and let
$z\in\B^n$ denote the estimated input to the topmost copy
of $f$. If $z$ is not in $\Dom(f)$, the algorithm $A_k$
declares failure and terminates (outputs $\bot$).
If $z\in\Dom(f)$, the algorithm finds a certificate
$c$ for $f$ consistent with $z$ of size at most $\C(f)$.
The certificate $c$ reveals at most $\C(f)$ bits.
For each of those bits, the algorithm $A_k$ then
applies $A_{k-1}$ on the corresponding copy of $f^{k-1}$,
repeating the algorithm $A_{k-1}$ if it outputs $\bot$
for a maximum of $1+\lceil\log \C(f)\rceil\le 2+\log\C(f)$
times to ensure the probability of failure drops to
at most $1/4\C(f)$. We then check all the resulting
certificates. If the function values of those copies of
$f^{k-1}$ agree with the predictions (from the certificate
$c$), we merge all the resulting certificates into
one big certificate for $f^k$ and return it
(this is a valid certificate since we've certified
all the bits of a certificate $c$ for the topmost copy
of $f$). Otherwise, if one of the function values
fails to match the bit predicted by the certificate $c$,
we return $\bot$. This completes the definition of $A_k$.

Note that $A_k$ has only a $1/4$ probability of generating
a wrong prediction $z$ for the topmost input to $f$.
Moreover, if $z$ is correct, then assuming $x\in\Dom(f^k)$,
we must have $z\in\Dom(f)$, and $c$ is a valid certificate
for $z$. In this case, the only way for the algorithm
to fail is if one of the runs of $A_{k-1}$ fails even
after the $1+\lceil\log\C(f)\rceil$ repetitions.
The probability that any one of these runs fails
for all those repetitions is at most $1/4\C(f)$,
and the probability that all $\C(f)$ of those runs
(for different bits of the certificate $c$) fail
is therefore at most $1/4$. By the union bound, the
probability that $A_k$ outputs $\bot$ is therefore at most
$1/2$, so it is a valid certificate-finding algorithm.

We now analyze its query complexity. In the first part,
when we find the guess string $z$, we use at most
$C\Q(f^{k-1})n\log n$ queries. In the second part,
we repeat $A_{k-1}$ for each of $\C(f)$ bits,
repeated for $2+\log\C(f)$ times each; the number
of queries used for this is
$q_{k-1}\C(f)(2+\log \C(f))$. This gives the recurrence
\[q_k\le C\Q(f^{k-1})n\log n+\C(f)(2+\log \C(f))q_{k-1}.\]
Let $a=Cn\log n$ and $b=\C(f)(2+\log \C(f))$.
Then 
\[q_k\le a\Q(f^{k-1})+bq_{k-1}\le
a\Q(f^{k-1})+ab\Q(f^{k-2})+b^2q_{k-2}\le\dots\qquad\]
\[\qquad\le
a\Q(f^{k-1})+ab\Q(f^{k-2})+ab^2\Q(f^{k-3})+\dots+ab^{k-2}\Q(f^1)+b^{k-1}q_1.\]
We now use $q_1=n\le a$ and $\Q(f^\ell)\le d\Q(f)^\ell$
for a constant $d\ge 1$ and all $\ell$.
This lets us bound $q_k$ by a geometric series:
\[q_k\le da\Q(f)^k\sum_{i=1}^{k-1}\left(\frac{b}{\Q(f)}\right)^i.\]
Now, if $b\ge \Q(f)/2$, then replacing $b$ with $4b$,
we get a geometric series that increases by a factor
of at least $2$ each time, and hence is upper bounded
by twice the largest term; in this case
$q_k\le 2da\Q(f) (4b)^{k-1}\le da(4b)^k$, or
\[\Q_{\C}(f^k)\le C'n\log n\cdot (5\C(f)\log (4\C(f)))^k\]
for $C'=dC$ (the $5$ comes from the addition of an extra
$\C(f^k)\le \C(f)^k$ factor that comes from
the definition of $\Q_{\C}(f^k)=q_k+\C(f^k)$).
Alternatively, if $b\le \Q(f)/2$, then
the geometric series decreases by a factor of $2$ each
time, and hence is upper bounded by twice the first term;
in this case
$q_k\le 2dab\Q(f)^{k-1}\le da\Q(f)^k$, or
\[\Q_{\C}(f^k)\le C'n\log n\cdot \Q(f)^k+\C(f)^k.\]
In both cases, we get
\[\Q_{\C}(f^k)\le C'n\log n\cdot
\left((5\C(f)\log 4\C(f))^k+ \Q(f)^k\right).\]
Taking both sides to the power $1/k$ and sending the limit
as $k\to\infty$, we get
\[\Q_{\C}^*(f)\le \max\{5\C(f)\log 4\C(f),\Q(f)\}.\]
Finally, we take the star of both sides. By
\cor{StarStar}, we have $\Q_{\C}^{**}(f)=\Q_{\C}^*(f)$.
By \lem{StarCompare}, taking the $*$ of the measure
$5\C(f)\log 4\C(f)$ yields $\C^*(f)$. From this it follows
that
\[\Q_{\C}^*(f)\le \max\{\Q^*(f),\C^*(f)\}.\]
The corresponding lower bound is easy: any algorithm
for finding a certificate for $f^k$ must have cost
at least $\C(f^k)$ by definition of certificate-finding
algorithms. Also, any certificate-finding algorithm can
easily be converted to a bounded-error algorithm at no extra
cost (simply guess the output randomly if the certificate
finding algorithm gave $\bot$ as output). This shows
the lower bounds of $\Q^*(f)$ and $\C^*(f)$,
as desired.

The case of randomized algorithms is exactly the same,
except for the final recurrence relation; there, we used
$\Q(f^\ell)\le d\Q(f)^\ell$ for a constant $d$, which
fails for randomized algorithms. Instead, we have
$\R(f^\ell)\le (d\R(f)\log \R(f))^\ell$ for a constant $d$.
The final bound then looks like
\[\R_{\C}(f^k)\le C'n\log n\cdot
\left((5\C(f)\log 4\C(f))^k+(d\R(f)\log \R(f))^k\right),\]
and the rest of the proof proceeds analogously
(first taking powers $1/k$ on both sides and taking
the limit $k\to\infty$ to get the upper bound
$O(\C(f)\log4\C(f)+\R(f)\log\R(f))$, and then
taking stars on both sides to improve this to the upper bound
$\max\{\C^*(f),\R^*(f)\}$).
We note that $\R_{\C}^*(f)=\R_0^*(f)$, since
$\R_{\C}(f)$ and $\R_0(f)$ differ by constant factors
(and using \lem{StarArithmetic}).
\end{proof}

\phantomsection\addcontentsline{toc}{section}{Acknowledgements}
\section*{Acknowledgements}
This research is supported in part by the Natural Sciences and
Engineering Research Council of Canada (NSERC), DGECR-2019-00027 and
RGPIN-2019-04804.\footnote{Cette recherche a été financée par le
Conseil de recherches en sciences naturelles et en génie du Canada
(CRSNG), DGECR-2019-00027 et RGPIN-2019-04804.}

\phantomsection\addcontentsline{toc}{section}{References} 
\renewcommand{\UrlFont}{\ttfamily\small}
\let\oldpath\path
\renewcommand{\path}[1]{\small\oldpath{#1}}
\emergencystretch=1em 
\printbibliography

\end{document}